\documentclass[sigconf]{acmart}

\usepackage{amsthm}
\newtheorem{theorem}{Theorem}

\usepackage{multirow}

\usepackage{algorithm}
\usepackage{algorithmicx}
\usepackage{algpseudocode}

\AtBeginDocument{%
  \providecommand\BibTeX{{%
    \normalfont B\kern-0.5em{\scshape i\kern-0.25em b}\kern-0.8em\TeX}}}

\begin{document}

\title{Automated Regression Unit Test Generation for Program Merges}

\author{Tao Ji}
\email{taoji@nudt.edu.cn}
\affiliation{%
  \institution{National University of Defense Technology}
  \city{Changsha}
  \state{Hunan, China}
  \postcode{410073}
}

\author{Liqian Chen}
\email{lqchen@nudt.edu.cn}
\affiliation{%
  \institution{National University of Defense Technology}
  \city{Changsha}
  \state{Hunan, China}
  \postcode{410073}
}

\author{Xiaoguang Mao}
\email{xgmao@nudt.edu.cn}
\affiliation{%
  \institution{National University of Defense Technology}
  \city{Changsha}
  \state{Hunan, China}
  \postcode{410073}
}

\author{Xin Yi}
\email{yixin_09@nudt.edu.cn}
\affiliation{%
  \institution{National University of Defense Technology}
  \city{Changsha}
  \state{Hunan, China}
  \postcode{410073}
}

\author{Jiahong Jiang}
\email{jhjiang@nudt.edu.cn}
\affiliation{%
  \institution{Beijing Institute of Tracking and Telecommunication Technology}
  \city{Beijing}
  \state{China}
  % \postcode{410073}
}
% \affiliation{%

% \affiliation{%

% \affiliation{%

% \affiliation{%

% \affiliation{%




\begin{abstract}
Merging other branches into the current working branch is common in collaborative software development.
However, developers still heavily rely on the textual merge tools to handle the complicated merge tasks.
The latent semantic merge conflicts may fail to be detected and degrade the software quality.
Regression testing is able to prevent regression faults and has been widely used in real-world software development.
However, the merged software may fail to be well examined by rerunning the existing whole test suite.
Intuitively, if the test suite fails to cover the changes of different branches at the same time, the merge conflicts would fail to be detected.
Recently, it has been proposed to conduct verification on 3-way merges, but this approach does not support even some common cases such as different changes made to different parts of the program.
In this paper, we propose an approach of regression unit test generation specifically for checking program merges according to our proposed test oracles.
And our general test oracles support us to examine not only 3-way merges, but also 2-way and octopus merges.
Considering the conflicts may arise in other locations besides changed methods of the project, we design an algorithm to select UUTs based on the dependency analysis of the whole project.
On this basis, we implement a tool called TOM to generate unit tests for Java program merges.
We also design the benchmark MCon4J consisting of 389 conflict 3-way merges and 389 conflict octopus merges to facilitate further studies on this topic.
The experimental results show that TOM finds 45 conflict 3-way merges and 87 conflicts octopus merges, while the verification based tool fails to work on MCon4J.
\end{abstract}






\maketitle

\section{Introduction}

Developers utilize the version control systems to make their own changes and accept contributions from other developers.
During the process, conflicts may arise once developers merge other branches into the current working branch.
Conflicts annoy the developers and developers have to take much efforts and be careful to deal with these conflicts~\cite{McKee17}.
After decades of hard working on detecting merge conflicts, various tools have been proposed~\cite{mens02,Brun11,Guimaraes12,Kasi13} to assist developers in detecting conflicts, reviewing changes and resolving conflicts.
However, considering the generality and usability of these tools, developers still rely on textual merge tools (e.g., those integrated in version control systems) to deal with their daily merge work~\cite{McKee17}.
In addition, a recent study conducted by Ahmed et al.~\cite{Ahmed17} shows that merges contain more code smells once conflicts arise.
This reality motivates us to figure out some general and effective method that has the potential to be widely used in real-world development.

Regression testing has been widely used to prevent regression faults and ensure the software quality after changes are made to the software.
Obviously, regression testing also can be used to detect conflicts after merging branches.
Once one test fails on the merge version but passes on all of the parent versions, this test reveals the merge conflicts.
Along this direction, Brun et al.~\cite{Brun11} classify merge conflicts into \textit{textual} and \textit{higher-order} conflicts (i.e., build and test conflicts), and tell whether test conflicts arise according to the results of rerunning the existing test suite.
As is well known, maintaining a high-quality test suite requires much efforts.
Hence, automatically generating unit test cases has been studied extensively. 
And, to evolve the test suite with the software, regression unit test generation is proposed to ensure that the software does not have other unexpected behaviors brought by new changes.

It seems that ideally rerunning the test suite is enough for detecting merge conflicts if we have a pretty high-quality test suite before merging branches.
However, in practice, the test suite may still have a high probability to miss the merge conflicts due to the work flow of the collaborative development.
Imagine that two branches are developed by different developers, the newly added or changed test cases may not well cover the changes made by the other branch since the developers are not aware of others' work.
In other words, if there exists one execution path that includes both changes of two branches, this path would not be covered by the existing test suite.
As a result, the conflicts between these changes fail to be detected by running the existing test suite.
If developers carefully examine the relationships between changes of different branches with the help of change impact analysis tools, the conflicts can be reduced.
However, a recent survey~\cite{Jiang2017} shows that developers do not use any change impact analysis tool in their daily debugging work, although they think these tools are much helpful.

Recently, Sousa et al.~\cite{Sousa18} have made progress on guaranteeing the quality of 3-way merges, by proposing the contract of \textit{semantic conflict freedom} and developing the tool SafeMerge to verify whether one 3-way merge meets this contract.
However, the state-of-the-art verification approach still has a set of limitations and challenges in verifying program merges.
First, besides 3-way merges, Git also supports to merge two branches without the common ancestor (i.e., 2-way merge) or more than 2 branches (i.e., octopus merge).\footnote{\url{https://git-scm.com/docs/git-merge/2.22.0}}
SafeMerge only supports 3-way merge ant thus fails to deal with other common merge scenarios.
Second, SafeMerge only works on those cases that two branches make changes to the same Java method.
Intuitively, changes made to different methods also may bring conflicts if these two methods are invoked along the same execution path.
Third, as described in~\cite{Sousa18}, SafeMerge has a set of limitations on changes made to method signatures, the analysis scope and exceptions.
In addition, considering the challenges for static analysis of Java reflection, verification on Java program merges involving Java dynamic features is not sound~\cite{Landman17}.

In this paper, we propose general test oracles for merges inspired from the contract of semantic conflict freedom and make the oracles applicable for all real-world merge scenarios (i.e., 2-way, 3-way and octopus merges).
After that, based on our proposed test oracles, we address the problems that how to find the UUTs (\textbf{U}nit \textbf{U}nder \textbf{T}esting) from the whole project and which variant involved in the merge scenario should be used to generate test cases.
We implement a tool named TOM (\textbf{T}esting \textbf{o}n \textbf{M}erges) to automatically find the impacted methods due to changes and then generate test cases to reveal conflicts.
Moreover, we construct a benchmark named MCon4J (\textbf{M}erge \textbf{Con}flicts \textbf{for} \textbf{J}ava).
MCon4J contains a total of 389 three-way merges and 389 octopus merges respectively, in each of which merge conflicts exist.
Then, we conduct experiments to examine the effectiveness of TOM.

In summary, our contributions are as follows:
\begin{itemize}
\item We propose to apply the regression unit test generation for detecting merge conflicts as the supplement of existing regression testing that does not consider merges specifically;
\item We propose the notion of general test oracles for checking semantic conflicts in merges, which supports 2-way, 3-way and octopus merges;
\item We design an algorithm to efficiently find the UUTs from the whole project based on dependency analysis such that we can concentrate on generating test cases to detect semantic conflicts over these UUTs;
\item We design the benchmark MCon4j consisting of 389 conflict 3-way merges and 389 conflict octopus merges, to facilitate further studies on detecting semantic conflicts; 
\item Experimental results show that our tool TOM finds 45 conflict 3-way merges and 87 octopus merges.
\end{itemize}

The rest of the paper is structured as follows. First we give our motivation (Section II). Then we detail our proposed approach (Section III) and evaluation (Section IV). After that, we present the discussion (Section V) and related work (Section VI). Finally, we conclude (Section VII).

\section{Motivation}

In this section, we discuss the shortcomings of existing methods that can be used to detect merge conflicts, which motivates and inspires our work.

\subsection{Verification on Three-way Merges} 

First, we recall the notion of semantic conflict freedom, proposed recently by Sousa et al.~\cite{Sousa18} for verifying 3-way merges.

\noindent\textit{Definition 2.1} (\textbf{Semantic Conflict Freedom})
Suppose that we are given four program versions
$O,A,B,M$ representing the base program, its two variants, and the merge respectively.
We say that $M$ is semantically conflict-free, if for all valuations $\sigma$ such that:

\centerline{$\sigma \vdash O \Downarrow \sigma_O,~\sigma \vdash A \Downarrow \sigma_A,~\sigma \vdash B \Downarrow \sigma_B,~\sigma \vdash M \Downarrow \sigma_M$}

\noindent the following conditions hold for all $i$ (where $i \in [0,len(out)], out$ represents the outputs): 

\noindent(1) If $\sigma_O$[($out$,$i$)]$\mathrm\neq\sigma_A$[($out$,$i$)], then $\sigma_M$[($out$,$i$)]=$\sigma_A$[($out$,$i$)] 

\noindent(2) If $\sigma_O$[($out$,$i$)]$\mathrm\neq\sigma_B$[($out$,$i$)], then $\sigma_M$[($out$,$i$)]=$\sigma_B$[($out$,$i$)] 

\noindent(3) Otherwise, $\sigma_O$[($out$,$i$)] = $\sigma_A$[($out$,$i$)] = $\sigma_B$[($out$,$i$)] = $\sigma_M$[($out$,$i$)]

Specifically, if one variable's value returned by variant \textit{A} (resp. \textit{B}) differs from the same variable's value returned by the base \textit{O}, then this variable's return value of the merge \textit{M} should agree with \textit{A} (resp. \textit{B}).

\begin{figure}
\centering
\includegraphics[width=0.7in]{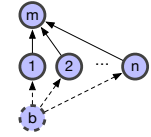}
\caption{The $n-parent$ merge has $n$ parents merged and these parents may have the common ancestor.}
\label{git_merge}
\end{figure}

\textbf{Real-world Merge Scenarios.} In real-world software development, Git is the most popular distributed version control system.
Different from the 3-way merge, 2-way merge has two branches merged without the common ancestor.
Git supports merging the other unrelated branch\footnote{\url{https://git-scm.com/docs/git-merge\#Documentation/git-merge.txt---allow-unrelated-histories}} into the working branch, which is one typical case of the 2-way merge.
Imagine that two branches both add methods with the same name compared to their ancestor.
Without any common ancestors, we also consider this merge as 2-way merge.
Moreover, Git supports merging multi-branches into the working branch, which is called the octopus merge.
However, we cannot use the notion of semantic conflict freedom in Definition 2.1 for verifying 2-way and octopus merges.
The meaning of \textit{\textbf{n-way}} merge has some ambiguity, for example, Le{\ss}enich et al~\cite{Lessenich2017} denote the octopus merge as n-way merge.
Hence, in our paper, we use the \textit{\textbf{n-parent}} merge to describe the real-world merge scenarios, as shown in Fig.~\ref{git_merge}.

\textbf{False Positives.}
In real-world merge scenarios, we cannot get the merge version successfully when two variants conflicts.
To have the final merge, developers may make a concession or try to figure out other ways.
However, the semantic conflict freedom is defined based on the idea that the merge should reserve all the changes of different branches.
Since other determination on the merge's behavior has been introduced when developers resolve conflicts, the notion of semantic conflict freedom does not fit.

% \begin{figure}
% \centering
% \includegraphics[width=3in]{motivating_example_1}
% \caption{Textual conflict arises in this merge scenario.}
% \label{example1}
% \end{figure}

According to the definition of semantic conflict freedom, once two variants change the same variable's value to different values, we can say the semantic conflict arises without analyzing the merge result.
Consider one example calculating one person's total income. 
The base version gets the income by ``income = salary'', the first variant changes it to ``income = salary + stock'' and the second changes it to ``income = salary + rent''.
Actually, two assignments do not conflict with each other considered developers' intention, and we shall use ``income = salary + stock + rent'' as the merge result.
In this case, this merge is not semantic conflict free while it may be the best merge candidate.
Moreover, the verification tool SafeMerge\footnote{Its latest commit is b2bac46ada till the date that we submit this paper, and we conduct experiments on this version.} reports conflicts without providing any counterexamples, which will increase developers' burden on checking the merge result.
Hence, we consider proposing the test generation approach to generate tests that cover the conflict parts.
Then, it will be much easier for developers to examine the merges.
\begin{figure}
\centering
\includegraphics[width=1.6in]{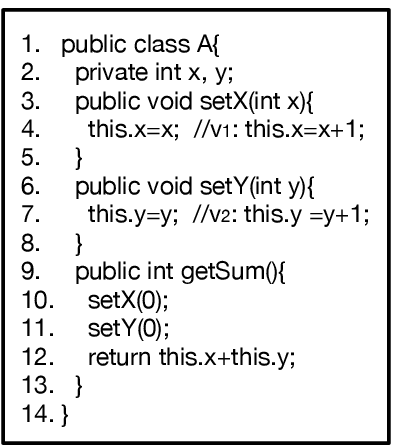}
\caption{Semantic conflicts arise in this merge.}
\label{example2}
\end{figure}

\textbf{False Negatives.} SafeMerge works at the method level and only works when both branches make changes to the same method.
Hence, we wonder whether semantic conflicts arise when two branches modify the different methods.
As shown in the Fig.~\ref{example2}, this Java class has two setter methods.
The first variant changes the body of \textit{setX} to ``this.x=x+1;'', while the second variant changes the body of \textit{setY} to ``this.y=y+1;''.
Obviously, we can merge them successfully by those textual merge tools, while SafeMerge will not report the semantic conflict for the method \textit{getSum}.
However, given the same input for four versions of \textit{getSum}, the returned results actually violate the semantic conflict freedom defined in Definition 2.1.
Imagine some more complicated cases that one method invokes changed methods indirectly.
If we want to improve SafeMerge to reduce the false negatives of above cases, we should analyze the invoked changed methods, which would increase the complexity of verification and may bring false positives due to the limitation of current verification techniques.

\subsection{Regression Testing}
Regression testing is used to make sure that the changes made to software are intended and do not introduce any unexpected behavior.
In the ideal case, developers of different branches add new tests to represent their intention on changes in both two branches, and after merging these two branches, developers can simply rerun the whole test suite to make sure the software works well without any failure.
However, in practice, since developers may be not aware of the other changes introduced by different branches in advance, the added test cases may not cover those changes.
As a result, rerunning the test suite may not expose the conflict introduced by the merged branch.

Existing works on automated regression unit test generation focus on the different parts between two versions.
And, to the best of our knowledge, there is no prior work aiming to automatically generate regression unit tests for the merge scenarios.
Obviously, it is relatively easy for developers to investigate the program behavior by examining test cases.
We wonder whether we can make use of the recent advance of verification on 3-way merges and regression testing to generate test cases to reveal conflicts for n-parent merges.

\section{Proposed Approach}

Similarly, as for the general problem of test generation, we need to figure out what to test and where to test.
In this section, we first propose our notion of test oracles to reveal merge conflicts.
Then, we present our algorithm for selecting the UUTs that may contain merge conflicts from the the whole program.
Finally, we detail the implementation of our tool TOM.

\subsection{Test Oracles}
\subsubsection{Test Oracles for 3-way Merges} In our testing based approach, we need to find the inputs that make the contract of semantic conflict freedom fail.
Before introducing our proposed test oracles, we first recall the algorithm~\cite{Sousa18} used in SafeMerge for verifying the semantic conflict freedom on 3-way merges.
As shown in Algorithm~\ref{alg:safemerge}, after computing the post relation on output variables of four versions in a 3-way merge (Line 2), the algorithm validates whether the logic formula is satisfiable (Line 6).

\begin{figure}[t]
\centering
\includegraphics[width=2in]{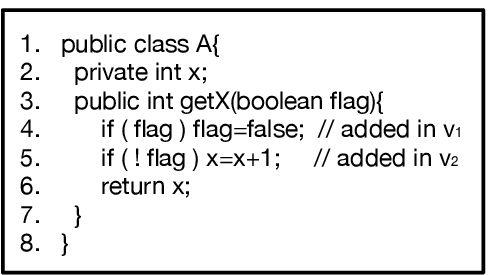}
\caption{Semantic conflicts that fail to be detected by SafeMerge.}
\label{error_example}
\end{figure}

\begin{algorithm}
\caption{Algorithm used in SafeMerge~\cite{Sousa18} for verifying semantic conflict freedom on 3-way merges}
\label{alg:safemerge}
\begin{algorithmic}[1]
\Function{verify}{$O$,$A$,$B$,$M$}
  \State $\varphi := RelationalPost(O,A,B,M)$
  \State $\chi_1 := \forall i. (out_O[i] \neq out_A[i] \Rightarrow out_A[i] = out_M[i]$)
  \State $\chi_2 := \forall i. (out_O[i] \neq out_B[i] \Rightarrow out_B[i] = out_M[i]$)
  \State $\chi_3 := \forall i. (out_O[i] = out_A[i] = out_B[i] = out_M[i]$)
  \State \Return $Valid(\varphi \Rightarrow (\chi_1\wedge\chi_2)\vee\chi_3)$
\EndFunction
\end{algorithmic}
\end{algorithm}

By examining this algorithm, we find that it is inconsistent with the definition of semantic conflict freedom (Definition 2.1).
For example, as shown in the Fig.~\ref{error_example}, the method \textit{getX} returns \textit{x} based on the input \textit{flag}.
If the \textit{flag} is set to \textit{true}, the merge version returns the different value compared to all of its three ancestors (i.e., the original version and two variants).
There is no doubt that this merge is not semantic conflict free according to the definition.
However, this case would not be reported if we adopt Algorithm~\ref{alg:safemerge} to verify the property.\footnote{We have examined the implementation of SafeMerge, and the same problem exists.}
We revise the algorithm into the following one:
\begin{algorithm}
\caption{Revised algorithm for verifying conflict freedom}
\label{alg:revised}
\begin{algorithmic}[1]
\Function{verify}{$O$,$A$,$B$,$M$}
  \State $\varphi := RelationalPost(O, A, B, M)$
  \State $\chi_1 := \forall i. (out_O[i] \neq out_A[i] \Rightarrow out_A[i] = out_M[i]$)
  \State $\chi_2 := \forall i. (out_O[i] \neq out_B[i] \Rightarrow out_B[i] = out_M[i]$)
  \State $\chi_3 := \forall i. (out_O[i]\mathrm{=} out_A[i] \mathrm{=} out_B[i] \Rightarrow out_O[i] = out_M[i]$)
  \State \Return $Valid(\varphi \Rightarrow (\chi_1\wedge\chi_2\wedge\chi_3))$
\EndFunction
\end{algorithmic}
\end{algorithm}

As shown in Algorithm~\ref{alg:revised}, if there exists one input that makes the logic formula \texttt{false} (Line 6), we would say the conflict exists.
To find the inputs revealing conflicts, we negate the $\chi_1\wedge\chi_2\wedge\chi_3$, then we have $\neg\chi_1\vee\neg\chi_2\vee\neg\chi_3$, where 

\centerline{$\neg\chi_1 \triangleq \exists i. (out_O[i] \neq out_A[i] \wedge out_A[i] \neq out_M[i])$}

\centerline{$\neg\chi_2 \triangleq \exists i. (out_O[i] \neq out_B[i] \wedge out_B[i] \neq out_M[i])$}

\centerline{$\neg\chi_3 \triangleq \exists i. (out_O[i] = out_A[i] = out_B[i] \wedge out_O[i] \neq out_M[i])$.}

Hence, to reveal merge conflicts, we should find the inputs $I$ that make the formula 

\centerline{$E(O,I) \wedge E(A,I) \wedge E(B,I) \wedge E(M,I) \Rightarrow \neg\chi_1\vee\neg\chi_2\vee\neg\chi_3$}

\noindent satisfiable, where $E(V,I)$ represents the execution result on the program $V$ by inputs $I$.
In other words, after execution on four different versions with the same inputs $I$, if $\neg\chi_1\vee\neg\chi_2\vee\neg\chi_3$ is satisfiable, we say the inputs $I$ reveal the merge conflict.
To simplify the problem of finding the inputs revealing conflicts, we give the following theorem to reformulate $\neg\chi_1\vee\neg\chi_2\vee\neg\chi_3$.
\begin{theorem}
$\neg\chi_1 \vee \neg\chi_2 \vee \neg\chi_3 \Leftrightarrow \neg\chi_1 \vee \neg\chi_2 \vee\neg\chi_3'$, where\\$\neg\chi_3' \triangleq \exists i. (out_A[i] \neq out_M[i] \wedge out_M[i] \neq out_B[i])$.
\end{theorem}
\begin{proof}
Let $O\triangleq out_O[i], A\triangleq out_A[i], B\triangleq out_B[i], M\triangleq out_M[i].$

\noindent$\exists i.(M\neq A \wedge A\neq O) \vee$ 
$\exists i.(M\neq B \wedge B\neq O) \vee $
$\exists i.(O = A = B  \wedge M\neq O \wedge M\neq A \wedge M\neq B)$

\noindent$\Leftrightarrow$ 
$\exists i.((M\neq A \wedge A\neq O) \vee$
$(M\neq B \wedge B\neq O) \vee$
$(O = A \wedge O = B \wedge M\neq A \wedge M \neq B))$\\
{//} $\exists i. P(i) \vee \exists i. Q(i) \Leftrightarrow \exists i.(P(i)\vee Q(i))$

\noindent $\Leftrightarrow$
$\exists i.((M\mathrm{\neq} A \wedge A\mathrm{\neq} O) \mathrm{\vee}$
$(M\mathrm{\neq} A \wedge A\neq O \wedge M\neq B \wedge B=O) \vee$
$(M\neq B \wedge B\neq O) \vee$
$(O = A \wedge O = B \wedge M\neq A \wedge M \neq B))$\\
{//} $P \Leftrightarrow P\vee (P \wedge Q)$

\noindent $\Leftrightarrow$
$\exists i.((M\mathrm{\neq} A \mathrm{\wedge} A\mathrm{\neq} O)\mathrm{\vee}$
$(M\mathrm{\neq} B \mathrm{\wedge} B\mathrm{\neq} O) \mathrm{\vee}$
$(O\mathrm{=}B \mathrm{\wedge} M\neq A \wedge M \neq B))$\\
{//} $(P\wedge Q)\vee (P\wedge \neg Q) \Leftrightarrow P$

\noindent $\Leftrightarrow$
$\exists i.((M\neq A \wedge A\neq O) \vee$
$(M\neq B \wedge B\neq O) \vee$\\
$(M\neq B \wedge B\neq O \wedge M\neq A) \vee$
$(O = B \wedge M\neq A \wedge M \neq B))$\\
{//} $P \Leftrightarrow P\vee (P \wedge Q)$

\noindent $\Leftrightarrow$
$\exists i.((M\neq A \wedge A\neq O) \vee$
$(M\neq B \wedge B\neq O) \vee$
$(M\neq A \wedge M\neq B))$\\
{//} $(P\wedge Q)\vee (P\wedge \neg Q) \Leftrightarrow P$

\noindent $\Leftrightarrow$
$\exists i.(M\neq A \wedge A\neq O) \vee$
$\exists i.(M\neq B \wedge B\neq O) \vee$
$\exists i.(M\neq A \wedge M\neq B)$
{//} $\exists i. P(i) \vee \exists i. Q(i) \Leftrightarrow \exists i.(P(i)\vee Q(i))$

\noindent Hence, we have $\neg\chi_1 \vee \neg\chi_2 \vee \neg\chi_3 \Leftrightarrow \neg\chi \vee \neg\chi_2 \vee\neg\chi_3'$, based on the above equation.
\end{proof}

The transformation guarantees that these three cases cover all cases that violate the contract of semantic conflict freedom.
Specifically, if (1) one variable's value returned by the merge is different from the values returned by two variants (i.e., $\neg\chi_3'$), or (2) one variable's value returned by one variant is different from that of the original and that of the merge (i.e., $\neg\chi_1$ or $\neg\chi_2$), the merge is not semantic conflict free.

\subsubsection{Generalized Test Oracles for n-parent Merges}
By exploring further $\neg\chi_1$, $\neg\chi_2$ and $\neg\chi_3'$, we can find the similarities between them.
In general, for any $\neg \chi_i \in [\neg\chi_1,\neg\chi_2,\neg\chi_3']$, given one version $v_t$, if we find that its output is different from those of two other variants $v_1$ and $v_2$, we say $\neg\chi_i$ is \texttt{true}.
To find the conflicts in one 3-way merge, we just need to repeatedly find one output of one version is different from those of the other versions.
And according to Theorem 1, the repeatable processes of finding outputs meet $\neg\chi_1, \neg\chi_2$ and $\neg\chi_3'$ respectively on different versions guarantee the quality of the final program merge.

Hence, based on the above observation, the framework of Algorithm~\ref{alg:revised} (i.e., the repeatable process) works for general n-parent merges.
In a 2-way merge scenario consisting of two variants $v_1$ and $v_2$ that do not have the common original, we can not tell what behavior is newly introduced by any of them.
However, we are able to investigate whether the merge has some new behavior that is not introduced by any of its two parents, just by examining the $\neg \chi_3'$.
As for the multiple variants merged in the octopus merge scenario, we just need to investigate each of the variants repeatedly, like for those two variants of 3-way merge by examining the $\neg \chi_1$ and $\neg \chi_2$.
Similarly, for the merge version in the octopus merge scenario, as shown in Fig.~\ref{n_parents}, we just need to compare those outputs returned by variants \{$v_1, v_2, ..., v_n$\} with that of the merge version $v_m$.
Now, we are able to tell whether semantic conflicts arise for all real-world merge scenarios.

\begin{figure}
\centering
\includegraphics[width=1in]{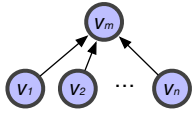}
\caption{$v_m$ is the merge version (i.e., the target program for generating tests) and [$v_1$, $v_2$, ..., $v_n$] are the variants (that need to fail the tests).}
\label{n_parents}
\end{figure}

Test case is a piece of code fragment including inputs and invoking UUTs.
Generally, assertions are often used to guard the values of those variables declared in the test case. 
If we have one test case generated for the version $v_m$ with the assertions on the values of output variables, then we can run this test case on all variants to check the failure of assertions.
If all variants fail on the same assertion, we tell that the behavior described by this assertion is missing in all variants.
Overall, the problem of finding inputs revealing conflicts can be transformed into the test generation problem.

\textit{Definition 3.1} (\textbf{Test Oracle on Program Merges})
Given an n-parent merge, to find the merge conflicts, we generate tests that achieve the following goals: 

\noindent(1) \textit{\textbf{Unexpected Behavior}}: Suppose that one test case $t$ is generated for the merged program $P_m$. We say that $P_m$ has some unexpected behavior, if for any parent version $P_i$, the same assertion $\varphi$ is violated;

\noindent(2) \textit{\textbf{Lost Behavior}}: Suppose that parent versions have the common original $P_o$. For one parent version $P_i$, we say its newly introduced behavior is missing after merging, if one test case $t$ for $P_i$ fails on $P_o$ and $P_m$ over the same assertion $\varphi$.

As we can see, our proposed test oracles make it in an easier way to detect conflicts, since we just need to repeatedly find one test case generated for one program version that can kill all of the variants by the same assertion.

\subsection{UUTs Selection}
As changes of different branches may be made to different locations across the whole program, we need to identity those methods whose behaviors have been affected by all branches.
For the general case shown in Fig.~\ref{n_parents}, we extract all sets of entities (i.e., fields and methods) that have different behavior between versions \{$\Delta(v_1,v_m)$, $\Delta(v_2,v_m)$, ..., $\Delta(v_n,v_m)$\}.
Once changes are made to one method, two versions of this method may have different behaviors.
And as shown in Fig.~\ref{example2}, the effects of changes also can be propagated into other unchanged locations.
Hence, if this changed method is called in the body of another method, the caller method may behave differently.
% Note that some changes like refactoring may not change the behavior at all.
% To prove the behavioral equivalence between program versions, much works \cite{Lahiri10,Partush14,Churchill19} have been done.
% However, considering the effectiveness, soundness and availability of these existing tools, we do not make use of them to rule out those caller methods that may not behave differently.
% Consequently, although we may need more resource to detect conflicts, we would not miss conflicts.

\begin{figure*}[t]
\centering
\includegraphics[width=6.6in]{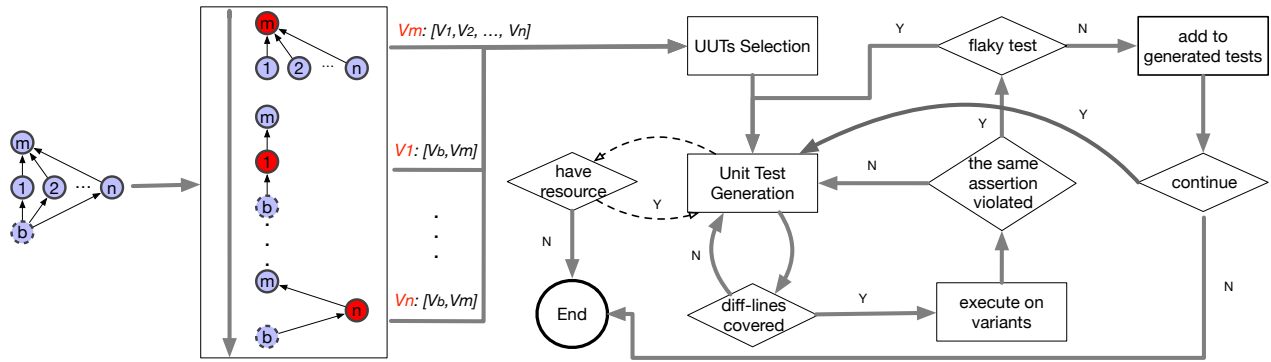}
\caption{The workflow of TOM.}
\label{workflow}
\end{figure*}

% In this paper, 
% Change impact analysis~\cite{Ren04,Buckner05,Dit14,Li13} is proposed to assist developers in examining other locations if changes have been made to some location.
% Similar to the change impact analysis tool Chianti~\cite{Ren04}, we extract the method-level dependencies to further analyze the change impact.
% Chianti defines a set of atomic changes including added method, changed method, deleted method, etc.
% Chianti is able to show developers where to modify once one method is deleted, whereas we do not need to find these locations.
In our paper, we extract those added and changed entities (i.e., fields and methods) and we do not consider the deleted class, field and method.
If one deleted method is not called by any methods, it seems meaningless to analyze this method.
If this deleted method was called by other methods and the developer has not deleted all of the invocations in other methods, the compiler would report errors.
As a result, to analyze the impact of deleted methods, we just need to analyze the modified methods.
The added entities should be included since these entities may be modified in other version pairs.
Imagine that in one 3-way merge, one method appears in the second parent versions $v_2$ and the merge version $v_m$, and these versions of this method are different.
This method in the merge should be analyzed since its behavior may be changed by developers after merging successfully.

Given one version $v_t$ together with a set of variants $\mathbb{V}$, after identifying added and changed entities from version pairs \{$\Delta(v_1,v_t)$, $\Delta(v_2,v_t)$, ..., $\Delta(v_n,v_t)$\}, we propose an algorithm for selecting UUTs from $v_m$ that behave differently in all variants, as shown in Algorithm~\ref{alg:uut}.
In some cases, the number of candidate methods may be large.
For example, if changes on the assignments of the fields are made to the constructor, the other methods which use the changed fields should be analyzed since we have to instantiate these methods before calling them.
Hence, one parameter $n$ is used to limit the size of all UUTs and one other parameter $d$ is used to limit the dependency depth explored.
The method-level dependency analysis is conducted on the target version to extract the dependency relationships between entities of the whole program (Line 2).
We then extract added and changed fields as well as methods (Lines 3-8) by comparing each variant with the target version.
During the exploration, we get more directly impacted entities by the already collected entities (Lines 11-19).
Then, we compute the intersection between the sets of impacted entities from each version pair as the UUTs (Line 20).
If we have a set of UUTs more than the number specified, we return the first $n$ UUTs to generate tests (Line 21).
If not, we find more impacted entities by exploring more deeply (Line 10).
If we fail to find the common impacted entities during the search (Lines 10-24), we generate test cases for those added and changed methods to ensure the program quality (Lines 25-26).

\begin{algorithm}
\caption{UUTs Selection}
\label{alg:uut}
\begin{algorithmic}[1]

\Function{select\_UUTs}{$v_t$, $\mathbb{V}$, $d$, $n$}

  \State $entity\_relations = extract\_dependencies(v_t)$
  \State $ce = \{\} $ {//} $ce$ is a list of sets of changed entities
  \State $ie = \{\} $ {//} $ie$ is a list of sets of impacted entities
  \For{ $v_i \in \mathbb{V}$}
    \State $ce[i] = diff(v_i, v_t) $
    \State $ie[i] = ie[i]\cup ce[i]$
  \EndFor
  \State $uuts=\emptyset$
  \For{$i \in [1,d]$}
    \For{$j\in [0, ie.size())$}
      \State $de = \emptyset$
      \For{$entity \in ie[index]$}
        \If{$entity.depth==i-1$}
        \State $de = de \cup get\_impacted(entity)$
        \EndIf
      \EndFor
      \State $ie[j] = ie[j] \cup de$
    \EndFor
    \State $uuts = intersect\_sets(ie)$
    \If{$uuts.size()>n$}
      \State \Return $uuts[0:n-1]$
    \EndIf
  \EndFor
  \If{$uuts == \emptyset$}
    \State $uuts=get\_all\_methods(ce)$
  \EndIf
  
\EndFunction

\end{algorithmic}
\end{algorithm}

\subsection{Implementation}

As shown in Fig.~\ref{workflow}, we present the work flow of our tool TOM.
Given one n-parent merge, we generate tests for \{$v_m$,$v_1$,...,$v_n$\} in order, as shown in the left part of Fig.~\ref{workflow}.
Then, for each case, we select the UUTs and then generate tests to reveal the conflicts.
We employ the dependency analysis tool \textit{depends}\footnote{\url{https://github.com/multilang-depends/depends}} to parse the whole program and generate the dependencies between different entities.

We employ the advanced test generation tool EvoSuite~\cite{Fraser11} to implement the the test generation for program merges.
% Various unit test generation tools~\cite{Pacheco07}\cite{AgitarOne}\cite{Fraser11} have been developed, and they can be roughly classified into the random test generation, search-based and dynamic symbolic execution tools.
EvoSuite utilizes meta-heuristic search techniques to generate and optimize test suites with respect to different code coverage criteria without reporting any false alarms~\cite{Gross12}.
EvoSuite generates test cases by different advanced techniques such as alternating variable method, random testing and dynamic symbolic execution.
% Hence, considering the features and the performance of EvoSuite, we decide to implement the test generation for program merges on top of EvoSuite.

During the search process of generating tests, coverage criteria are used to generate the high-quality test suit.
As for the UUTs in our case, it may have the same code in all variants and the target version, while the different parts are located in other methods, classes or packages.
Existing implemented coverage criteria in EvoSuite instruments the code of the UUT, and then generate coverage goals for the UUT.
High coverage on the UUT still cannot guarantee that the different parts are covered.
Hence, we implement the diff-line coverage criterion to guide the search to achieve coverage goals on different lines between two versions.

Considering that the execution is expensive, based on the extracted different lines, we determine whether to execute the generated test case on all different variants to detect conflicts.
Given two program versions with the test case $t$, if their execution results are different, we can tell that the different parts between those two versions must be covered.
Hence, as shown in Fig.~\ref{workflow}, if one generated test can not cover any different part between two versions, we do not need to re-execute the test on the other version.
If the generated test can cover the different parts, we then execute it on all of the variants.

For each execution of the test case on different variants, EvoSuite is able to collect the values of those variables in the test case.
Given any two executions, we are able to generate assertions on the variables that have different values to capture the behavior difference between two versions.
After running all the variants, we extract all the assertions that appear in each execution comparison.
These assertions are what we need to reveal the conflicts.
If one test case triggers some exception for the target version, we leave it to developers since the exception may be not desired.
For each statement that has exceptions thrown in the execution on the variant, we generate all assertions for this statement based on the execution on the target version to describe the different values and states.

If we have the same assertion generated by executions on all variants, we then check the stability of this test by rerunning the tests five times.
If we find one stable test, we add it to the test list that will be provided to developers to examine the merge conflicts.
As shown in Fig.~\ref{workflow}, we can configure TOM to stop the unit test generation once one test case revealing conflicts is generated or the given resource has been consumed (e.g., time out).

\section{Experimental Evaluation}
In this section, we describe the details of constructing the benchmark MCon4J, and conduct experiments to evaluate the effectiveness of TOM. 

\subsection{Benchmark}

Sousa et al.~\cite{Sousa18} collect a total of 52 merges from nine real-world open-source Java applications.
In our evaluation, we do not reuse these merges to conduct the experiments, because (1) changes of different branches are made to the same method in these 3-way merges; (2) the false positives and the false negatives may exist in their results, we are not able to ensure the conflict exists without knowing the counterexamples revealing the conflict; and (3) our tool's effectiveness on detecting semantic conflicts heavily depends on the unit test generation tool.
Although unit test generation techniques have been extensively improved, the advanced tools still fail to generate the tests that achieve high coverage for real-world programs consisting of complex objects, structures and logics.
Honestly, the static verification techniques can outperform the dynamic unit test generation techniques in some aspects such as the time costs.
Given the above considerations, we do not conduct experiments on these 52 merges collected by Sousa et al~\cite{Sousa18}.

To evaluate our tool's effectiveness on detecting semantic conflicts, we need one benchmark consisting of 3-way merges and octopus merges.
And for each merge, we need to know whether the conflicts exist or not.
Otherwise we cannot tell that the tool works well if the testing based tool fails to find the conflict.
And considering that changes made to different parts also may bring conflicts, we need to evaluate our tool on those merges whose parents modify the different methods.
However, unlike the merges with textual conflicts, the merges that have semantic conflicts are pretty difficult to be identified by examining the evolutionary history of real-world programs.
Hence, we decide to construct conflict 3-way and octopus merges by leveraging real-world bug-fix activities.

Brun et al.~\cite{Brun11} classifies merge conflicts into \textit{textual} and \textit{higher-order} conflicts.
The higher-order conflicts arise when changes are semantically incompatible and cause compilation errors, test failures, etc.
Based on this notion of merge conflict, we construct conflict merges by making bug-fix tests fail.
If we have one branch created for fixing bugs, but the test case fails after merging other branches, we consider conflict exists.
The remaining problem is how to create other branches.
Just et al.~\cite{Just14:mutants} conduct one investigation of the coupling effect between real faults and the mutants that are generated by commonly used mutation operators. 
The results show the existence of a coupling effect for 73\% of real faults.
Hence, we decide to use the generated mutants as the other branches along with the bug-fix branch to construct 3-way and octopus merges to conduct experiments.

\textbf{Construction of 3-way merges.}
We first execute the bug-fix test case on the buggy program to collect the covered lines that may include those of unchanged classes during the bug fixes.
Then, we make use of the mutation tool major~\cite{Just2014:Major} to generate mutants on those covered lines of the buggy program.
For each buggy program $P_b$, we have the bug-fixed version $P_f$ and the mutant version $P_t$.
Using the default recursive merge strategy of Git, we construct merges that may have semantic conflicts by merging the $P_f$ and $P_t$.
For each buggy program, we have a number of merges.
Then, we execute the bug-fix test case on the merged program $P_m$, we get one conflict merge if this test case fails.

\textbf{Construction of octopus merges.}
For each constructed 3-way merge, we create another branch by randomly choose one mutant of the same class mutated in the $P_t$.
Then, we merge these three branches by using the default merge strategy of Git.
Similarly, the bug-fix test case will fail on the merge program.

Just et al.~\cite{Just14} propose Defects4J which collects a total of 438 reproducible bugs\footnote{https://github.com/rjust/defects4j, the latest commit is 17a99e1} to facilitate other software engineering research.
For example, Defects4J is widely used in different fields such as unit test generation~\cite{Shamshiri15}, automated program repair~\cite{Martinez2017}.
For some versions of the Mockito project, Defects4J fails to generate mutants\footnote{https://github.com/rjust/defects4j/issues/198. The issue remains open and has not been resolved till the date that we submit this paper}.
And for some bugs collected in Defects4J, we fail to construct conflict 3-way merges by merging mutants to make the bug-fix test cases fail.
Finally, we collect a total of 389 conflict three-way merges and 389 conflict octopus merges.

As shown in Table \ref{benchmark}, we list the numbers of merges in each project.
The column ``DC'' means that the mutate branch and the fix branch modify the different classes. 
The column ``SC'' means that two branches modify the different methods of the same class.
The column ``SM'' means that two branches modify the same method.

\begin{table}[H]
\centering
\caption{The numbers of 3-way merges and octopus merges. }
\label{benchmark}
\begin{tabular}{|c|r|r|r|r|r|}
\hline
\multirow{2}{*}{Project} & \multicolumn{1}{c|}{\multirow{2}{*}{Bugs}} & \multicolumn{4}{c|}{Merges}  \\ \cline{3-6} & \multicolumn{1}{c|}{}  & \multicolumn{1}{c|}{DC} & \multicolumn{1}{c|}{SC} & \multicolumn{1}{c|}{SM} & \multicolumn{1}{c|}{Total} \\ \hline
JFreeChart & 26 & 18 & 6  & 1  & 25   \\ \hline
Closure  & 176  & 168  & 3  & 0  & 171 \\ \hline
Commons Lang  & 65  & 16  & 34  & 11  & 61   \\ \hline
Commons Math  & 106  & 61 & 32  & 6   & 99   \\ \hline
Mockito  & 38 & 6 & 0 & 0 & 6  \\ \hline
Joda-Time & 27  & 23  & 4  & 0 & 27   \\ \hline
Total & 438  & 292  & 79 & 18 & 389  \\ \hline
\end{tabular}
\end{table}

% \subsection{Research Questions}
% After conducting experiments, we would like to answer two main research questions:
% \begin{itemize}
% \item \textbf{RQ1}: How many 3-way conflict merges are found by SafeMerge and TOM respectively?
% \item \textbf{RQ2}: How does TOM perform on those constructed octopus merges?
% \end{itemize}
% \noindent The answers of these two research questions are able to show the effectiveness and limitations of TOM, and guide the further works on detecting conflict merges.

\subsection{RQ1: How many conflict 3-way merges are found by SafeMerge and Tom respectively?}

\textbf{\textit{SafeMerge.}} Once changes of branches are made to different methods, SafeMerge fails to verify whether conflicts exist.
As shown in Table I, there are 18 three-way merges in which modifications are made to the same method.
We run SafeMerge on these 18 merges to evaluate its effectiveness.
However, SafeMerge fails to deal with 17 out of 18 merges by throwing the same error\footnote{This issue has been reported to the authors while it has not been fixed till the date that we submit this paper. We would update the results once the issue is resolved.} which stops the verification procedure.
As for the remaining merge, SafeMerge fails to return results after running for one hour, which is much greater than the average time cost (i.e., most of them are less than one second and the greatest one is 4.45s) as shown in~\cite{Sousa18}.
Unfortunately, we fail to utilize SafeMerge to detect conflicts on all of those merges from MCon4J.

\textbf{\textit{TOM.}}
Recall that for selecting UUTs we limit the explored depth and the total number of UUTs.
In our experiments, we set the explored depth to 5 and the number of UUTs to 3 respectively.
We conduct two groups of experiments guided by different coverage criteria.
For the first experiment, we only use the proposed diff-line coverage criterion.
For the second experiment, we add more coverage criteria\footnote{The used criteria implemented in EvoSuite includes branch, cbranch, weakmutation, output, exception, method and methodnoexception.} used in EvoSuite by default.
Because the random operators existing in the search process, we run TOM on each 3-way merge three times to have a comprehensive view on the tool's ability. 

\begin{figure}[t]
\centering
\includegraphics[width=3.3in]{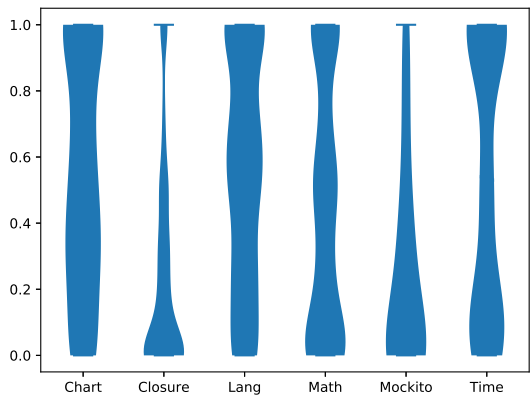}
\caption{The distributions of coverage information.}
\label{coverage_info}
\end{figure}

As shown in Table \ref{results_3way}, we list the total of 45 conflict merges detected by TOM.
There are 42 and 40 conflict 3-way merges detected respectively by two groups of experiments.
We fail to tell the multiple criteria work better than the diff-line criterion by examining these experimental results.

\begin{table*}[t]
\caption{Conflict three-way merges that detected by executing three times with different coverage criteria ``Diff-line'' and ``Multi-criteria''. ``$-$'' means the conflict merge fails to be detected. ``$\odot$'' means the conflict merge is detected when generating test cases for the $P_f$ variant (i.e., the bug-fix version). ``$\ominus$'' means the conflict merge is detected when generating test cases for non-$P_f$ variants. ``$\oplus$'' means the conflict merge is detected when generating test cases for both $P_f$ and non-$P_f$ variants.}
\label{results_3way}
\centering
\begin{tabular}{|c|c|c|c|c|c|c|c||c|c|c|c|c|c|c|c|}
\hline
\multirow{2}{*}{Id} & \multirow{2}{*}{Type} & \multicolumn{3}{c|}{Diff-line} & \multicolumn{3}{c||}{Multi-criteria} & \multirow{2}{*}{Id} & \multirow{2}{*}{Type} & \multicolumn{3}{c|}{Diff-line} & \multicolumn{3}{c|}{Multi-criteria} \\ \cline{3-8} \cline{11-16} 
 &  & \#1  & \#2 & \#3 & \#1 & \#2 & \#3 & & & \#1 & \#2 & \#3 & \#1 & \#2 & \#3 \\ \hline
 Chart\_4 &DC& $\odot$ & $\odot$ & $\odot$ & $\odot$ & $\odot$ & $\odot$ & Math\_16 &SC& $-$ & $-$ & $\odot$ & $-$ & $\odot$ & $-$ \\ \hline
Chart\_5 &DC& $-$ & $\odot$ & $-$ & $-$ & $-$ & $-$ & Math\_27 &DC& $\odot$ & $-$ & $\odot$ & $\odot$ & $\odot$ & $-$ \\ \hline
Chart\_8 &SC& $\ominus$ & $\oplus$ & $\ominus$ & $\ominus$ & $\ominus$ & $\ominus$ & Math\_29 &DC& $\odot$ & $\odot$ & $\odot$ & $\odot$ & $\odot$ & $\odot$ \\ \hline
Chart\_11 &SC& $-$ & $\oplus$ & $\ominus$ & $\ominus$ & $\odot$ & $\oplus$ & Math\_32 &DC& $-$ & $\odot$ & $-$ & $-$ & $-$ & $\odot$ \\ \hline
Chart\_14 &DC& $-$ & $-$ & $-$ & $-$ & $-$ & $\odot$ & Math\_37 &DC& $\oplus$ & $-$ & $-$ & $-$ & $-$ & $-$ \\ \hline
Chart\_16 &SC& $\oplus$ & $\oplus$ & $\oplus$ & $\oplus$ & $\oplus$ & $\oplus$ & Math\_46 &SC& $-$ & $\odot$ & $-$ & $-$ & $\oplus$ & $\oplus$ \\ \hline
Chart\_21 &DC& $\oplus$ & $\oplus$ & $\oplus$ & $\oplus$ & $\oplus$ & $\oplus$ & Math\_47 &SC& $\oplus$ & $\ominus$ & $\ominus$ & $\oplus$ & $\oplus$ & $\oplus$ \\ \hline
Chart\_24 &SC& $\oplus$ & $\odot$ & $\oplus$ & $\oplus$ & $\oplus$ & $\oplus$ & Math\_49 &DC& $\odot$ & $\odot$ & $\odot$ & $\odot$ & $\odot$ & $\odot$ \\ \hline
Closure\_19 &DC& $\odot$ & $\odot$ & $\odot$ & $\odot$ & $\odot$ & $\odot$ & Math\_56 &SC& $\odot$ & $\odot$ & $\odot$ & $\odot$ & $\odot$ & $\odot$ \\ \hline
Closure\_173 &DC& $-$ & $\ominus$ & $\ominus$ & $-$ & $-$ & $-$ & Math\_60 &DC& $\odot$ & $\odot$ & $\odot$ & $-$ & $-$ & $\oplus$ \\ \hline
Lang\_19 &DC& $-$ & $-$ & $-$ & $\odot$ & $-$ & $-$ & Math\_63 &SC& $\oplus$ & $\odot$ & $-$ & $\oplus$ & $\oplus$ & $\oplus$ \\ \hline
Lang\_39 &SC& $\odot$ & $\oplus$ & $\oplus$ & $\oplus$ & $-$ & $\oplus$ & Math\_70 &SC& $-$ & $\odot$ & $\odot$ & $\odot$ & $\odot$ & $\odot$ \\ \hline
Lang\_41 &SC& $\odot$ & $\odot$ & $\odot$ & $\odot$ & $\odot$ & $\odot$ & Math\_71 &DC& $\odot$ & $\odot$ & $\odot$ & $\odot$ & $\odot$ & $-$ \\ \hline
Lang\_45 &DC& $\oplus$ & $\oplus$ & $\oplus$ & $\oplus$ & $\oplus$ & $\oplus$ & Math\_73 &SC& $\ominus$ & $\ominus$ & $\ominus$ & $\ominus$ & $\ominus$ & $\ominus$ \\ \hline
Lang\_47 &SC& $\odot$ & $\odot$ & $\odot$ & $\odot$ & $\odot$ & $\odot$ & Math\_80 &SC& $\odot$ & $\odot$ & $-$ & $-$ & $\odot$ & $\odot$ \\ \hline
Lang\_60 &SC& $\oplus$ & $\oplus$ & $\oplus$ & $\oplus$ & $\oplus$ & $\oplus$ & Math\_81 &SC& $\oplus$ & $\odot$ & $\odot$ & $\odot$ & $-$ & $\odot$ \\ \hline
Lang\_65 &SM& $\oplus$ & $\oplus$ & $\oplus$ & $\oplus$ & $\oplus$ & $\oplus$ & Math\_83 &DC& $\odot$ & $\odot$ & $\odot$ & $\odot$ & $\odot$ & $\odot$ \\ \hline
Math\_1 &DC& $\odot$ & $-$ & $\odot$ & $-$ & $-$ & $\odot$ & Math\_92 &SM& $-$ & $\oplus$ & $\oplus$ & $\oplus$ & $\ominus$ & $\oplus$ \\ \hline
Math\_2 &DC& $-$ & $-$ & $\odot$ & $-$ & $-$ & $-$ & Math\_93 &SC& $\odot$ & $\odot$ & $-$ & $\odot$ & $\odot$ & $\odot$ \\ \hline
Math\_4 &DC& $\odot$ & $\oplus$ & $\oplus$ & $\oplus$ & $\oplus$ & $\oplus$ & Math\_95 &SC& $\oplus$ & $\oplus$ & $\oplus$ & $\oplus$ & $\oplus$ & $\oplus$ \\ \hline
Math\_9 &SC& $-$ & $-$ & $-$ & $\odot$ & $-$ & $-$ & Math\_97 &DC& $\odot$ & $\odot$ & $-$ & $\odot$ & $-$ & $\odot$ \\ \hline
Math\_10 &DC& $\odot$ & $\odot$ & $\odot$ & $-$ & $-$ & $-$ & Time\_9 &DC& $\ominus$ & $-$ & $\ominus$ & $\ominus$ & $\ominus$ & $\ominus$ \\ \hline
Math\_11 &DC& $-$ & $\odot$ & $-$ & $-$ & $\odot$ & $\odot$ & & & & & & & & \\ \hline
\end{tabular}
\end{table*}

During the process of generating test cases for merges, if we can find one execution path cover the changes, we may detect the merge conflict.
In other words, higher coverage does not always mean the conflicts can be found.
However, higher coverage still improves the possibility of detecting conflicts.
As shown in the Fig.~\ref{coverage_info}, we present the achieved maximum coverage during the test case generation for each target $P_2$ variant (i.e., the bug fixed version) guided by the diff-line coverage.
As we can see, TOM achieves lower coverage on the projects Closure and Mockito than the four other projects.
The existing study~\cite{Shamshiri15} also shows that EvoSuite achieves the low coverage when generates test cases for the Closure project (the Mockito project is not included in the Defects4J at that time).
It seems that the relatively high coverage is achieved for the project Time, while there are more flaky tests generated according to our observations on the test generation process.
And these flaky tests affects the coverage information collected during the unit test generation.

\subsection{RQ2: How does TOM perform on those constructed octopus merges?}
We construct octopus merges based on constructed 3-way merges.
To answer this question, we adapt the same settings and use the set of multiple criteria described in RQ1.
Then, we run TOM on each octopus merge three times.

As shown in Table \ref{results_octopus}, we show the details of experimental detection results on octopus merges.
There are a total of 87 conflict octopus merges detected by TOM.
Comparing those detected conflict 3-way merges, we find those octopus merges from 35 out of 45 cases whose 3-way conflict merges have been detected.
A total of 52 conflict octopus merges from newly appeared cases are detected.
We construct octopus merges by adding one mutated branch based on the constructed 3-way merge.
With more mutated code changing program behaviors, it makes sense that TOM can find more conflict octopus merges than conflict 3-way merges.

\begin{table*}
\caption{Conflict octopus merges detected by executing three times with the multiple coverage criteria. The same symbols are explained in the caption of Table II.}
\label{results_octopus}
\centering
\begin{tabular}{|c|c|c|c|c||c|c|c|c|c||c|c|c|c|c|}
\hline
Id      & Type & \#1     & \#2 & \#3 & Id & Type & \#1 & \#2 & \#3 & Id & Type & \#1 & \#2 & \#3 \\ \hline
Chart\_2 &DC & $\odot$ & $\odot$ & $\odot$ & Closure\_97 &DC & $\ominus$ & $\ominus$ & $\ominus$ & Math\_31 &DC & $\odot$ & $-$ & $-$ \\ \hline
Chart\_4 &DC & $\odot$ & $\odot$ & $\odot$ & Closure\_108 &DC & $\ominus$ & $\ominus$ & $\ominus$ & Math\_32 &DC & $\odot$ & $\odot$ & $-$ \\ \hline
Chart\_7 &DC & $\odot$ & $\odot$ & $-$ & Closure\_111 &DC & $-$ & $-$ & $\ominus$ & Math\_37 &DC & $\odot$ & $\oplus$ & $\odot$ \\ \hline
Chart\_11 &SC & $\odot$ & $\odot$ & $\odot$ & Closure\_138 &DC & $\odot$ & $\odot$ & $\odot$ & Math\_39 &DC & $-$ & $-$ & $\ominus$ \\ \hline
Chart\_14 &DC & $-$ & $\odot$ & $-$ & Closure\_140 &DC & $\ominus$ & $\ominus$ & $\ominus$ & Math\_46 &SC & $\odot$ & $\odot$ & $\odot$ \\ \hline
Chart\_20 &DC & $-$ & $\oplus$ & $\ominus$ & Closure\_148 &DC & $\odot$ & $\odot$ & $\odot$ & Math\_47 &SC & $\oplus$ & $\oplus$ & $\odot$ \\ \hline
Chart\_21 &DC & $\odot$ & $-$ & $-$ & Closure\_152 &DC & $\oplus$ & $\oplus$ & $\oplus$ & Math\_49 &DC & $\odot$ & $\odot$ & $\odot$ \\ \hline
Chart\_24 &SC & $\odot$ & $\odot$ & $\odot$ & Closure\_165 &DC & $\odot$ & $\odot$ & $\odot$ & Math\_52 &DC & $\oplus$ & $\oplus$ & $\ominus$ \\ \hline
Chart\_25 &DC & $\ominus$ & $\oplus$ & $\oplus$ & Lang\_15 &DC & $\oplus$ & $\oplus$ & $\ominus$ & Math\_56 &SC & $\odot$ & $\odot$ & $\odot$ \\ \hline
Closure\_1 &DC & $\ominus$ & $\odot$ & $-$ & Lang\_17 &DC & $\ominus$ & $\ominus$ & $\ominus$ & Math\_60 &DC & $\odot$ & $\odot$ & $\odot$ \\ \hline
Closure\_2 &DC & $\odot$ & $-$ & $\odot$ & Lang\_29 &DC & $\oplus$ & $\oplus$ & $\oplus$ & Math\_63 &SC & $\odot$ & $\odot$ & $\odot$ \\ \hline
Closure\_3 &DC & $\oplus$ & $\odot$ & $\odot$ & Lang\_34 &DC & $-$ & $\ominus$ & $-$ & Math\_70 &SC & $\oplus$ & $\odot$ & $\odot$ \\ \hline
Closure\_15 &DC & $\ominus$ & $\ominus$ & $\ominus$ & Lang\_36 &SC & $\odot$ & $\odot$ & $\odot$ & Math\_71 &DC & $\odot$ & $-$ & $\odot$ \\ \hline
Closure\_19 &DC & $\oplus$ & $\oplus$ & $\oplus$ & Lang\_37 &SC & $-$ & $\ominus$ & $\ominus$ & Math\_73 &SC & $-$ & $\ominus$ & $-$ \\ \hline
Closure\_24 &DC & $-$ & $\ominus$ & $\ominus$ & Lang\_38 &SC & $\odot$ & $\oplus$ & $\odot$ & Math\_80 &SC & $-$ & $\odot$ & $\oplus$ \\ \hline
Closure\_26 &DC & $-$ & $\ominus$ & $-$ & Lang\_39 &SC & $\odot$ & $\odot$ & $\odot$ & Math\_81 &SC & $\odot$ & $-$ & $\odot$ \\ \hline
Closure\_27 &DC & $\ominus$ & $\ominus$ & $\ominus$ & Lang\_41 &SC & $\odot$ & $\odot$ & $\odot$ & Math\_83 &DC & $\odot$ & $\odot$ & $\odot$ \\ \hline
Closure\_32 &DC & $\ominus$ & $\ominus$ & $\ominus$ & Lang\_45 &DC & $\odot$ & $\odot$ & $\odot$ & Math\_85 &DC & $\ominus$ & $\ominus$ & $\ominus$ \\ \hline
Closure\_47 &DC & $-$ & $-$ & $\odot$ & Lang\_47 &SC & $\oplus$ & $\oplus$ & $\oplus$ & Math\_87 &DC & $-$ & $\odot$ & $-$ \\ \hline
Closure\_55 &DC & $-$ & $\ominus$ & $\ominus$ & Lang\_60 &SC & $\odot$ & $\odot$ & $\odot$ & Math\_88 &DC & $\ominus$ & $\ominus$ & $\ominus$ \\ \hline
Closure\_67 &DC & $\odot$ & $\odot$ & $\odot$ & Lang\_65 &SM & $\odot$ & $\odot$ & $\odot$ & Math\_92 &SM & $\odot$ & $\odot$ & $\odot$ \\ \hline
Closure\_72 &DC & $\ominus$ & $\ominus$ & $\ominus$ & Math\_1 &DC & $-$ & $\odot$ & $-$ & Math\_93 &SC & $\odot$ & $\odot$ & $\odot$ \\ \hline
Closure\_78 &DC & $-$ & $\ominus$ & $\ominus$ & Math\_4 &DC & $\odot$ & $\odot$ & $\oplus$ & Math\_95 &SC & $\odot$ & $\odot$ & $\odot$ \\ \hline
Closure\_80 &DC & $\odot$ & $\odot$ & $\odot$ & Math\_9 &SC & $\ominus$ & $\oplus$ & $\ominus$ & Math\_96 &SC & $\oplus$ & $\oplus$ & $\oplus$ \\ \hline
Closure\_81 &DC & $\ominus$ & $\ominus$ & $\ominus$ & Math\_11 &DC & $\odot$ & $\odot$ & $-$ & Math\_102 &DC & $\odot$ & $\odot$ & $\odot$ \\ \hline
Closure\_84 &DC & $-$ & $\ominus$ & $\ominus$ & Math\_16 &SC & $\odot$ & $-$ & $-$ & Math\_103 &DC & $\ominus$ & $-$ & $-$ \\ \hline
Closure\_89 &DC & $\ominus$ & $-$ & $\ominus$ & Math\_25 &SC & $-$ & $-$ & $\ominus$ & Math\_104 &SC & $\odot$ & $-$ & $-$ \\ \hline
Closure\_95 &DC & $\ominus$ & $\ominus$ & $\ominus$ & Math\_27 &DC & $\odot$ & $\odot$ & $\odot$ & Time\_9 &DC & $\oplus$ & $\odot$ & $\odot$ \\ \hline
Closure\_96 &DC & $\ominus$ & $-$ & $-$ & Math\_29 &DC & $\odot$ & $\odot$ & $\odot$ & Time\_20 &DC & $-$ & $-$ & $\ominus$ \\ \hline
\end{tabular}
\end{table*}

% \subsection{RQ3}

% % Please add the following required packages to your document preamble:
% % \usepackage{multirow}
% \begin{table}[]
% \begin{tabular}{|c|r|r|r|r|r|r|}
% \hline
% \multirow{2}{*}{Project} & \multicolumn{2}{c|}{\#1} & \multicolumn{2}{c|}{\#2} & \multicolumn{2}{c|}{\#3}  \\ \cline{2-7} \multicolumn{1}{|c|}{} & zero  & partial  & zero  & partial & zero & paritial \\ \hline
% \multicolumn{1}{|c|}{Chart} &  &  &  &  &  &  \\ \hline
% \multicolumn{1}{|c|}{Closure} &  &  &  &  &  &  \\ \hline
% \multicolumn{1}{|c|}{Lang} &  &  &  &  &  &  \\ \hline
% \multicolumn{1}{|c|}{Math} &  &  &  &  &  &  \\ \hline
% \multicolumn{1}{|c|}{Mockito} &  &  &  &  &  &  \\ \hline
% \multicolumn{1}{|c|}{Time} &  &  &  &  &  &  \\ \hline

% \end{tabular}
% \end{table}

% As shown above, we have detected conflicts for three-way merges and octopus merges.
% As for those merges that we fail to detect conflicts, we collect the coverage information from the generated test suites.
% As shown in the Table , we. 

\subsection{Threats to Validity}
The main threats to the validity of our results belong to the internal and external validity threat categories.

Internal validity threats correspond to the implementation of TOM and the relevant scripts.
Although we have reviewed the implementation carefully, the bugs may still exist and threat to the validity of our results.

External validity threats correspond to the constructed benchmark MCon4J.
We mutate the source code by a set of mutant operators to simulate the real-world changes.
Even though mutants can be used as good resource for research.
Our tool may fail to achieve the similarly good results for detecting conflicts on real-world merge scenarios.
The parameters affect the UUTs selection, and the values used also may fail to work well for other merge scenarios.
As a result, these parameters may be adjusted in other merge scenarios.

\section{Discussion}

In this section, we discuss merge conflict resolution, the application scenario of TOM and the potential techniques on improving TOM.

\textbf{Merge Conflict Resolution.}
% The textual merge tools are the most widely used in the real-world merge scenarios.
% If different developers make changes to the same line, the tool would report conflicts that need to be resolved manually.
% Operation-based merging techniques~\cite{Lippe92} try to reuse the change operations of one branch to similar locations of the other branch.
Besides the textual tools, other existing syntactical and structural tools~\cite{mens02}\cite{Zhu2018} also aim to keep all changes introduced by different branches in the merge version.
However, the reasonableness of keeping and combining all changes may be questionable according to the contract of semantic conflict freedom.
When different changes are made to the same assignment statement, the values of the variables would be different.
In this case, whatever developers have done to resolve conflicts, the values of this variable in different versions would violate the contract of semantic conflict freedom.
Without knowing other requirements on the merge, it is unable to resolve the conflicts automatically.
For example, to resolve conflicts automatically, Xing and Maruyama~\cite{Xing19} develops one automated program repair tool which needs one test case representing contracts provided by developers.
Existing empirical studies~\cite{Menezes18} on the resolutions of merge conflicts also show that developers are much likely to simply choose one version as the merge result when conflict arises.
Hence, instead of replacing the textual tools with fine-grained tools to generate merges automatically, we think that we should pay more attention on guaranteeing the software quality after merging.

\textbf{Continuous Integration.}
Continuous integration service is widely used in the open-source community to automatically find the build or test failures.
Considering the workflow of collaborative development and the relative expensive costs of generating test cases for detecting conflicts, we think integrating the unit test generation in continuous integration is practical.
After rerunning the whole test-suite, the continuous integration service is able to collect the coverage information.
Then, if no tests fail, it utilizes the existing test case that covers changed code to seed the unit test generation to accelerate the procedure of generating test cases to trigger the semantic conflicts. 
Given the resource limit or the coverage requirements, the continuous integrations service is helpful to give developers more confidence on the quality of the merged software.

\textbf{Potential Improvement.}
Besides the ability of test generation tools, identifying those methods that may have latent conflicts also plays an important role in the test generation for detecting conflicts.
In our paper, we select the UUTs for generating tests based on the explicit call and use dependencies between changed fields, constructors and methods.
However, we may miss some dependencies between changed entities, which leads to the failure of detecting conflicts.
Those dependencies include the co-change relationships that can be mined from the software evolutionary history, the documented API usages and other common contracts (e.g., developers should call the \textit{close} method to free resource after invoking \textit{open} method).
In real-world merge scenarios, we can seed the existing test cases that represents some dependency between entities to generate tests.
Still, the missing of dependencies between methods has an impact on the detection results.

During the procedure of selecting UUTs, we extract all of the changed entities.
Note that some changes like refactoring may not change the behavior at all.
To prove the behavioral equivalence between program versions, much works \cite{Lahiri10,Partush14,Churchill19} have been done.
However, considering the effectiveness, soundness and availability of these existing tools, we currently do not make use of them to rule out those change impacted methods that may not behave differently.
In future, with the help of advanced tools that can precisely determine the different behavior introduced by changes, we can save some computing resource by filtering out those changed methods whose behaviors are not changed.

\section{Related Work}
In this section, we mainly describe the related works on software merging and regression testing.

\subsection{Software Merging}
Over decades, software merging has been extensively studied because of its important role in software maintenance and evolution.
Ahmed et al.~\cite{Ahmed17} study the relationship between code smells and merge conflicts, and results show merges contain more code smells when conflicts arise.
Nearly twenty years ago, Mens~\cite{mens02} provided a comprehensive summary of excellent works in this field.
Different merging techniques such as textual, syntactic, semantic, structural and operation-based merging have been proposed very early.

However, the situation that developers rely on textual merge tools to deal with their daily work has not been changed.
Mckee et al.~\cite{McKee17} conduct interviews of 10 software practitioners to understand their perspectives on merge conflicts and resolutions. 
According to the unmet needs of software practitioners, they suggest researchers and tool builders focus on program comprehension, history exploration, etc. 
Nishimura and Maruyama~\cite{Nishimura16} present one tool that exploits the fine-grained edit history to assist developers to examine the merge conflicts.

For the last decade, some new ideas and trends have emerged.
Considering the variety of program languages, semi-structural merging~\cite{Apel11} aims to achieve the balance between generality and performance.
Proactive or early detection~\cite{Brun11}\cite{Guimaraes12}\cite{Nguyen15} of conflicts is used to decrease the possibility of merging branches with conflicts.
Providing a set of candidate conflict resolutions to developers is also helpful.
Niu et al.~\cite{Nan12} develop a tool scoreRec that recommends the conflict resolutions ordered by estimating the cost and benefit of resolving conflicts.
Zhu and He~\cite{Zhu2018} propose an interactive approach that ranks the conflict resolutions generated via the version space algebra.
Xing and Maruyama~\cite{Xing19} introduce the automatic program repair techniques to resolve the merge conflicts by leveraging the existing test cases.

Sousa et al.~\cite{Sousa18} propose the contract of semantic conflict freedom inspired from much earlier work~\cite{Horwitz89}\cite{Yang90}, and then propose the verification on three-way merges to increase developers' confidence on the merge result with respect to the contract. 
In this paper, we propose the test oracle inspired from the semantic conflict freedom and make it applicable for all merge scenarios.
Utilizing the state-of-the-art unit test generation tool, we can generate tests to reveal conflicts if any.

\subsection{Regression Testing}

Regression test selection and regression test generation are the major techniques trying to prevent regression faults effectively with the low cost.
The cost of rerunning the whole test suite is growing with the size and complexity of the evolving software.
Various regression test selection techniques have been proposed.
Zhu et al.~\cite{Zhu2019} propose the framework for checking the regression test selection tool.

Different from selection, regression test generation aims to generate new test cases that can expose the behavioral differences between two versions, when existing test suites fail to expose difference.
Taneja and Xie~\cite{Taneja08} synthesize one driver that calls two versions of the target method and adds conditions comparing the return values.
Then, they utilize existing tools to generate test cases that cover the different branches in the driver method.  
As Jin et al.~\cite{Jin10} explain, the generated test cases may not reveal the regression faults while cover the changed parts.
Hence, they develop BERT to generate test cases that cover different parts first, and then analyze the behavioral differences to reveal the regression faults.

Person et al.~\cite{Person08} propose differential symbolic execution that leverages symbolic execution techniques to characterize the changes, without providing the inputs to execute the changed program.
Person et al.~\cite{Person11} propose directed incremental symbolic execution to find those path conditions affected by code changes. 
Taneja et al.~\cite{Taneja11} develop eXpress as one search strategy for dynamic symbolic execution to prune out paths that cannot lead to any code regions and those paths through which a state infection cannot propagate to any observable output.

Xu et al.~\cite{Xu10} propose directed test suite augmentation that identifies code affected by changes and existing test cases relevant to testing that code. 
Then the identified test cases are used to seed the concolic or genetic test case generation approach to create new test cases that execute the affected code.

Software merging acts as the important activity during the software evolution, whereas there is no tool aiming to generate test cases revealing conflicts after merging.
In our paper, we implement TOM to generate test cases for 2-way, 3-way and octopus merges.

\section{Conclusion}
In this paper, we propose the general test oracles for real-world program merges including 2-way, 3-way and octopus merges.
Based on our test oracles, we propose an approach of regression unit test generation for detecting semantic conflicts.
On this basis, we implement a tool called TOM to automatically generate test cases to reveal merge conflicts.
In addition, we design the benchmark MCon4J to support further studies on regression test generation for software merges.
In our experiments, a total of 45 conflict 3-way merges and 87 conflict octopus merges are identified by our tool TOM, while the state-of-the-art verification based tool SafeMerge fails to work on MCon4J.
The experimental results show that our regression unit test generation tool is useful and effective in guaranteeing the quality of real-world program merges.

\bibliographystyle{ACM-Reference-Format}
\bibliography{../../references}


\end{document}